\documentclass[conference]{IEEEtran}

\usepackage{pcctlstarmc}

\begin{document}
\title{LTL Fragments are Hard for\\ Standard Parameterisations}

\author{\IEEEauthorblockA{Martin Lück and Arne Meier}
\IEEEauthorblockA{Institut für Theoretische Informatik\\
Leibniz Universität Hannover\\
Appelstr.~4, 30176 Hannover, Germany\\
Email: \url|{lueck,meier}@thi.uni-hannover.de|}}

\maketitle

\begin{abstract}
We classify the complexity of the LTL satisfiability and model checking problems for several standard parameterisations.
The investigated parameters are temporal depth, number of propositional variables and formula treewidth, resp., pathwidth. We show that all operator fragments of LTL under the investigated parameterisations are intractable in the sense of parameterised complexity.
\end{abstract}


\section{Introduction}
In the last decade the research on parameterised complexity of problems increased significantly. Beyond the foundations by Downey and Fellows \cite{dofe99} until today several deep algorithmic techniques have been introduced and new approaches have been made; so it really is a highly prospering area of research (e.g., see for an overview of the current evolution the recent book of Downey and Fellows \cite{df13}).  Essentially the main approach is to detect a parameterisation for a given problem and achieve a runtime which is independent of the parameter. For instance, given the problem of propositional satisfiability a quite naïve parameterisation is the number of variables (of the given formula $\phi$). Then one can easily construct a deterministic algorithm solving the problem in time $2^{k}\cdot|\phi|$ where $k$ is the number of variables in $\phi$. If $k$ is assumed to be fixed then this yields a polynomial runtime wherefore one says that this problem is \emph{fixed-parameter tractable} ($\FPT$). In general, a problem is in $\FPT$ if there exists a deterministic algorithm solving the problem for any instance in $f(k)\cdot\textit{poly}(n)$ steps where $k$ is the parameter, $n$ is the input length, and $f$ is an arbitrary computable function; another name for this class is para-P. In contrast with this, runtimes of the form $n^{f(k)}$ are highly undesirable as the runtime depends on the parameter's value---this is the runtime of algorithms in the class XP. Further some kind of parameterised intractability hierarchy is built between FPT and XP, namely the W-hierarchy. It is known that $\FPT\subseteq\W1\subseteq\W2\subseteq\cdots\subseteq\XP$ but not if any of these inclusions is strict. For proving hardness results with respect to the classes of the W-hierarchy one usually uses fpt-reductions which translate, informally speaking, the instances in the usual sense from one parameterised problem to another but also take care of maintaining the parameter's value. Hence showing $\W1$-hardness of a problem yields the unlikeliness of it to be fixed-parameter tractable. Also classes like para-NP---similarly to para-P but using non-deterministic algorithms instead---or para-PSPACE have been introduced. The first contains the W-hierarchy and is by itself contained in para-PSPACE. Further they are incomparable to XP---under reasonable complexity class inclusion assumptions.

While the parameterised complexity theory is heavily built on top of logic, its application is relatively new in the context of logic itself. Not many significant parameterisations are known yet. Recent approaches were modal, resp., temporal operator nesting depth \cite{prav13, pc-ctl15} or various types of treewidth, like primal or incidence treewidth of CNF formulas \cite{ss06}. Our treewidth parameter is a further generalisation and can be measured on general syntax trees of formulas and not only on CNFs.

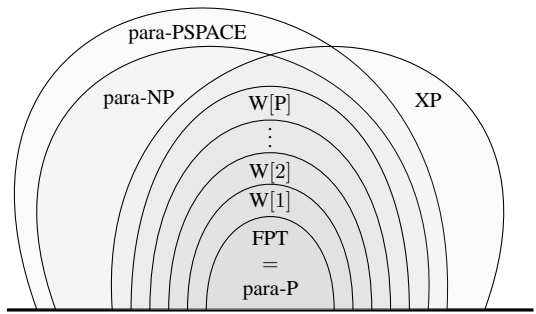
\begin{figure}
\centering
\begin{tikzpicture}[y=.35cm,x=.5cm, every node/.style={color=black,fill opacity=1}]


\draw[very thick] (7,0) -- (-7,0);

\clip (-7,-.1) rectangle (7,12);

{

\tikzset{every path/.style={draw=black,fill=black!20,fill opacity=0.1}}

\path (-6.2,0) edge [in=87,out=110,looseness=2.6] (4.7,0);

\path (-5.7,0) edge [in=86,out=110,looseness=2.5] (4.2,0);

\path (-4.2,0) edge [in=70,out=94,looseness=2.5] (5.7,0);

\path (-3.7,0) edge [in=90,out=90,looseness=2.75] (3.7,0);

\path (-3.2,0) edge [in=90,out=90,looseness=2.7] (3.2,0);

\path (-2.7,0) edge [in=90,out=90,looseness=2.65] (2.7,0);

\path (-2.2,0) edge [in=90,out=90,looseness=2.6] (2.2,0);

\path (-1.7,0) edge [in=90,out=90,looseness=2.5] (1.7,0);

}

\node[scale=0.8] at (0,1.7) {
	\begin{tabular}{c}
	$\FPT$\\
	$=$\\
	para-$\P$
	\end{tabular}
};
\node[scale=0.8] at (0,4.1) {$\W{1}$};
\node[scale=0.8] at (0,5.2) {$\W{2}$};
\node[scale=0.8] at (0,6.8) {$\vdots$};
\node[scale=0.8] at (0,7.8) {$\W{\P}$};
\node[scale=0.8] at (-3.5,8) {para-$\NP$};
\node[scale=0.8] at (-2.2,10.5) {para-$\PSPACE$};
\node[scale=0.8] at (4.2,8) {$\XP$};

\end{tikzpicture}
\caption{Parameterised complexity classes.}\label{fig:paraCompClasses}
\end{figure}

Temporal logic is a well-known and very important area relevant in many fields of research, e.g., program verification and artificial intelligence. Its origins are traceable even to greek philosophers yet it founds its introduction by Arthur Prior in the 1960s \cite{pr57}. Conceptually its main ingredients are combinations of a universal or existential \emph{path quantifier} together with \emph{temporal operators} referring to specific or vague moments of time, e.g., \emph{next}, \emph{globally}, \emph{future}, \emph{until}. Depending on how these quantifiers and operators may be combined the three most important logics have been defined. In \emph{Computation Tree Logic} (or short CTL) one is allowed to use only pairs of a single path quantifier and a single temporal operator; in \emph{Linear Temporal Logic} (LTL) one uses only temporal operators and has a single existential (or universal, depending on the definition of the logic) path quantifier in front of the formula; in the \emph{Full Branching Time Logic} (CTL$^{*}$) any arbitrary combination is allowed. After a decade of seminal work from Allen Emerson, Clarke, Halpern, Schnoebelen, and Sistla \cite{clem81,emha85,sicl85,sch02}---to name only a few---the most important concepts, e.g., satisfiability and model checking have been well understood and classified with respect to their computational complexity. Recently the mentioned decision problems have been investigated in the light of a study which considers fragments of the logics in the sense of allowed operators \cite{bamuscscscvo11,mmtv09}. 

In this paper we focus on the logic LTL and its $\PSPACE$-complete model checking. We want to investigate its parameterised complexity under the mentioned parameterisations of operator fragments.

\subsection{Related Work} The parameterised complexity of modal logic satisfiability has been investigated by Praveen recently \cite{prav13}. He considered treewidth of some CNF-centered graph representation structure. In a work of Szeider from 2004 he discusses different parameterisation approaches with respect to propositional satisfiability \cite{sz04}. In particular, he explains how to obtain primal graphs and other structural parameterisations. Recently, \citeauthor{pc-ctl15} classified the parameterised complexity of satisfiability for the computation tree logic CTL \cite{pc-ctl15}. Essentially we follow the used parameterisations from \citeauthor{pc-ctl15} in this paper. The results comply with the results of \citeauthor{bamuscscscvo11} who investigated the existential version of LTL \cite{bamuscscscvo11}. Outside the context of parameterised complexity theory, many simple cases of LTL were studied by \citeauthor{demri_complexity_2002} \cite{demri_complexity_2002}.

\subsection{Organization of this paper} At first we will define the relevant notions and terms around temporal logic, parameterised complexity, and our used structural parameterisations. Then in Section~\ref{sec:results} we will present our classification for the different parameterisations and decision problems. We start with satisfiability and the main part is about model checking. Finally in Section~\ref{sec:conclusion} we will conclude and give an outlook on further steps.

\section{Preliminaries}

\subsection{Temporal Logic}
Usually temporal logic is defined through Kripke semantics. In the following we will briefly introduce the notion around it. For a deeper introduction, we refer the reader to the survey article from \citeauthor{mmmv12} \cite{mmmv12}. Formally a \emph{Kripke structure} $\calA=(W,R,V)$ is a finite set $W$ of worlds (or states), $R\colon W\times W$ is a total transition relation (i.e., for every $w\in W$ there exists a $w'\in W$ such that $wRw'$ holds), and $V\colon W\to 2^{\PROP}$ is a valuation function assigning sets of propositions to states, where $\PROP$ is a finite set of propositions. A \emph{path} $\pi=p_{0}p_{1}\dots$ is an infinite sequence of states such that $p_{i}Rp_{i+1}$ holds for $i\in\mathbb N$. A path starting at the root of a model is also called a \emph{run}. For $w \in W$ write $\Pi(w)$ for the set of all paths starting in $w$. Write $\pi[i]$ for the world $p_{i}$ and $\pi^{i}$ for the suffix path $p_{i}p_{i+1}\dots$ of $\pi$.
The set of all well formed linear temporal logic formulas $\LTLform$ is defined inductively via the following grammar in BNF
\[
\varphi\ddfn 
p\mid
\lnot\varphi\mid
\varphi\land\varphi\mid
\X\varphi\mid
\F\varphi\mid
\varphi\U\varphi,
\]
where $p\in\PROP$. 

Now let $\calA=(W,R,V)$ be a Kripke structure, $\varphi,\psi$ be two $\LTLform$-formulas, and $\pi$ be a path in $\calA$. Then it holds that
$$
\begin{array}{lcl}
 (\calA,\pi)\models p &\text{iff}& p\in V(\pi[0]),\\
 (\calA,\pi)\models \lnot\varphi &\text{iff}& (\calA,\pi)\not\models \varphi,\\
 (\calA,\pi)\models \varphi\land\psi &\text{iff}& (\calA,\pi)\models \varphi\text{ and }(\calA,\pi)\models\psi,\\ 
 (\calA,\pi)\models \X\varphi &\text{iff}& (\calA,\pi^{1})\models\varphi,\\
 (\calA,\pi)\models \F\varphi &\text{iff}& \exists k\ge 0:(\calA,\pi^{k})\models\varphi,\\
 (\calA,\pi)\models \varphi\U\psi &\text{iff}& \exists k\ge 0\text{ such that } \forall j<k:\\
 &&(\calA,\pi^{j})\models\varphi\text{ and }(\calA,\pi^{k})\models\psi.
\end{array}
$$

The usual shortcuts are obtained by combinations of these operators, e.g., $\G\varphi\equiv\lnot\F\lnot\varphi$, or $\varphi\to\psi\equiv\lnot\varphi\lor\psi$. If $T\subseteq\{\X,\F,\G,\U\}$ is a set of temporal operators, then $\LTLform(T)$ is the restriction of $\LTLform$ to formulas containing only temporal operators from $T$. In the following we consider the relevant decision problems with respect to a fix set of temporal operators $T$ for this paper.

\problemdef
{$\LTLSAT(T)$}
{$\LTLform(T)$-formula $\varphi$}
{Does there exist a Kripke structure $\calA$ and a path $\pi$ in $\calA$ such that $(\calA,\pi)\models\varphi$?}

\problemdef
{$\LTLMC(T)$}
{A Kripke structure $\calA$, a world $w\in W$ and $\LTLform(T)$-formula $\varphi$}
{$(\calA,\pi)\models\varphi$ for all paths $\pi \in \Pi(w)$?}

Note that the model checking problem is also sometimes defined with \emph{existential} semantics, \ie, at least one path has to satisfy the formula, which leads to complementary complexity results. 

\begin{define}
The \emph{temporal depth} $\td(\cdot)$ of an LTL formula is defined recursively:
\begin{alignat*}{4}
  &\td(x) \; & &:= 0 \text{ if } x \in \PROP, &  & \td(\neg \varphi) \; & &:= \td(\varphi),\\
  &\td(\varphi \wedge \psi) \; & &:= \max\{\td(\varphi), \td(\psi)\}, & \qquad &  & &\\ 
  &\td(T \varphi) \; & &:= \td(\varphi) + 1 \qquad \mathrlap{\text{ for } T \in \{ \X, \F, \G \},}& \qquad & & &\\
  &\td(\varphi \U \psi) \; & &:= \max\{\td(\varphi), \td(\psi)\}+1. & \qquad &  & &
\end{alignat*}
\end{define}

\begin{define}
Define $\LTLlim{c}(T)$ as the fragment of $\LTLform(T)$ which has temporal depth at most $c$. Define $\LTLMClim{c}(T)$, resp., $\LTLSATlim{c}(T)$ as the model checking, resp., satisfiability problem restricted to formulas in $\LTLlim{c}(T)$.
\end{define}

\subsection{Parameterised Complexity}

\begin{define}[Parameterised problem]
Let $Q \subseteq \Sigma^*$ be a decision problem and let $\kappa \colon \Sigma^* \to \N$ be a computable function. Then we call $\kappa$ a \emph{parameterisation of $Q$} and the pair $\Pi = (Q, \kappa)$ a \emph{parameterised problem}.
\end{define}

\begin{define}[Fixed-parameter tractable]
Let $\Pi = (Q, \kappa)$ be a parameterised problem. If there is a deterministic Turing machine $M$ and a computable function $f\colon \N \to \N$ \st for every instance $x \in \Sigma^*$
\begin{itemize}
  \item $M$ decides correctly if $x \in Q$, and
  \item $M$ has a runtime bounded by $f(\kappa(x)) \cdot \size{x}^{\bigO{1}}$
\end{itemize}
then we say that $M$ is an \emph{fpt-algorithm for $\Pi$} and that $\Pi$ is \emph{fixed-parameter tractable}.
We define $\FPT$ as the class of all parameterised problems that are fixed-parameter tractable.
\end{define}

Similarly, we refer to a function $f$ as \emph{fpt-computable \wrt a parameter $\kappa$} if there is another computable function $h$ such that $f(x)$ can be computed in time $h(\kappa(x)) \cdot \size{x}^\bigO{1}$.

\begin{define}[fpt-reduction]
Let $(P, \kappa)$ and $(Q, \lambda)$ be parameterised problems over alphabets $\Sigma$, resp., $\Delta$. Then a function $f \colon \Sigma^* \to \Delta^*$ is an \emph{fpt-reduction} if it is fpt-computable \wrt $\kappa$ and there is a computable function $h \colon \N \to \N$ \st the following holds \fa $x \in \Sigma^*$:

\begin{itemize}
  \item $x \in P \iff f(x) \in Q$ and
  \item $\lambda\big(f(x)\big) \leq h\big(\kappa(x)\big)$, \ie, $\lambda$ is bounded by $\kappa$.
\end{itemize}

If there is an fpt-reduction from $(P, \kappa)$ to $(Q, \lambda)$ for parameterised problems $(P, \kappa)$ and $(Q, \lambda)$ then we call $(P, \kappa)$ \emph{fpt-reducible to} $(Q, \lambda)$, denoted by $(P, \kappa) \leqfpt (Q, \lambda)$.
\end{define}


\begin{define}[\W{1}]
The class $\W{1}$ is the class of parameterised problems $(Q, \kappa)$ such that $(Q, \kappa)$ can be fpt-reduced to the \emph{Short Single-Tape Turing Machine Halting Problem}:
\paraproblemdef
{SSTMH}
{Non-deterministic single-tape Turing machine $M$, Integer $k$}
{Does $M$ accept the empty string in at most $k$ steps?}
{$k$}
\end{define}

\begin{define}[Parameterised hardness]
A problem $(P, \kappa)$ is $\mathcal{C}$-\emph{hard under fpt-reduc\-tions} for a parameterised complexity class $\mathcal{C}$ if $(Q, \lambda) \in \mathcal{C}$ implies $(Q, \lambda) \leqfpt (P, \kappa)$. If additionally $(P, \kappa) \in \mathcal{C}$, we say that $(P, \kappa)$ is $\mathcal{C}$\emph{-complete under fpt-reductions}.
\end{define}

\begin{define}[$\W{\P}$]\label{def:wp}
The class $\W{\P}$ contains the parameterised problems $(Q,\kappa)$ for which there is a computable function $f$ and an NTM deciding if $x \in Q$ holds in time $f(\kappa(x)) \cdot \size{x}^\bigO{1}$ with at most $\bigO{\kappa(x) \cdot \log \size{x}}$ non-deterministic steps.
\end{define}

\citeauthor{flumGroheParaC} state how to obtain parameterised variants of classical complexity classes \cite{flumGroheParaC}.
They define for ``standard'' complexity classes $\calC$ the corresponding parameterised versions para-$\calC$. Here, ``standard'' means that the class $\calC$ is defined via usual resource-restricted Turing machines. 
For most such classes $\calC$ we obtain para-$\calC$ by simply appending an additional factor $f(\kappa)$ to the resource bound, as done for $\P$ leading to $\FPT$. This is possible for certain classes that \citeauthor{flumGroheParaC} call \emph{robust}, such as $\NP$ and $\PSPACE$. This allows the following definitions:

\begin{define}[para-$\NP$]
The class para-$\NP$ contains the parameterised problems $(Q,\kappa)$ for which there is a computable function $f$ and an non-deterministic Turing machine deciding if $x \in Q$ holds in time $f(\kappa(x)) \cdot \size{x}^\bigO{1}$.
\end{define}

\begin{define}[para-$\PSPACE$]
The class para-$\PSPACE$ contains the parameterised problems $(Q,\kappa)$ for which there is a computable function $f$ and a deterministic Turing machine deciding if $x \in Q$ holds in space $f(\kappa(x)) \cdot \size{x}^\bigO{1}$.
\end{define}

\begin{define}[Slice]
The $\ell$-th \emph{slice} of a parameterised problem $(Q, \kappa)$ is denoted with $(Q, \kappa)_\ell$ and defined as:
\[
(Q, \kappa)_\ell \dfn \Set{ x \mid x \in Q \text{ and } \kappa(x) = \ell }
\]
\end{define}

\begin{proposition}[\cite{flumGroheParaC}]
\label{thm:slice-hardness}For \emph{$\calC \in \{ \NP, \coNP, \PSPACE \}$} it holds that a parameterised
problem $(Q, \kappa)$ is hard for \emph{para-}$\calC$ if and only if a finite union
of slices of $(Q, \kappa)$ is hard for $\calC$.
\end{proposition}

\begin{define}[Complement Class]
For a complexity class $\calC$, define co$\calC$ as the class of problems for which their complement is in $\calC$. 
\end{define}

\subsection{Structural Treewidth and Pathwidth}
The \emph{treewidth} of a graph is a parameter that leads to $\FPT$ algorithms for a wide range of otherwise intractable graph problems. In fact, only few known graph problems stay hard on trees. The treewidth of a graph is in this sense a measure of its ``tree-alikeness''.

A \emph{path-decomposition} $P$ of a structure $\mathcal A$ is similarly defined to tree-decompositions however $P$ has to be a path. Here $\pw(\mathcal A)$ denotes the \emph{pathwidth} of $\mathcal A$. Likewise the size of the pathwidth describes the similarity of a structure to a path. Observe that pathwidth bounds treewidth from above.

Given a structure $\mathcal A$ we define a \emph{tree-decomposition of $\mathcal A$} (with universe $A$) to be a pair $(T,X)$ where $X=\{B_{1},\dots,B_{r}\}$ is a family of subsets of $A$ (the set of \emph{bags}), and $T$ is a tree whose nodes are the bags $B_{i}$ satisfying the following conditions:
\begin{enumerate}
 \item Every element of the universe appears in at least one bag: $\bigcup X=A$.
 \item Every Tuple is contained in a bag: for each $(a_{1},\dots,a_{k})\in R$ where $R$ is a relation in $\mathcal A$, there exists a $B\in X$ such that $\{a_{1},\dots,a_{k}\}\in B$.
 \item For every element $a$ the set of bags containing $a$ is connected: for all $a\in A$ the set $\{B\mid a\in B\}$ forms a connected subtree in $T$. (For \emph{path-decompositions} is has to form a connected path.)
\end{enumerate}

\begin{define}[Graph treewidth and pathwidth]
The \emph{width} of a tree-decomposition $\calT$ is its maximal bag size minus one. The \emph{treewidth} of a graph $G$ is the minimal width of a tree-decomposition of $G$, and its \emph{pathwidth} is the minimal width of a path-decomposition.
\end{define}

\begin{define}[Syntactical structure of formulas]
We associate a formula $\varphi$ with a graph $\calS(\varphi)$, resp., $\calS_\varphi$ which represents the formula.
The map $\calS$ is defined as follows: Let $\calS(\varphi) \dfn (\SFstar{\varphi}, E)$. The set $\SF{\varphi}$ is the
set of all syntactically valid subformulas of $\varphi$ (counting equal subformulas multiple times if necessary). $\SFstar{\varphi}$ is then obtained from $\SF{\varphi}$ by identifying nodes which represent the same propositional variable inside $\varphi$. The edge set $E$ is obtained from $\varphi$ by connecting each pair $(\psi,\psi') \in \SFstar{\varphi}\times\SFstar{\varphi}$ for which $\psi'$ is a maximal strict subformula inside $\psi$.
\end{define}

We assume a well-defined association with parentheses \st every node of $\calS(\varphi)$ represents exactly one Boolean function or temporal operator and its children represent its arguments.
Then $\calS$ can be interpreted as a ``graphical'' representation of $\varphi$ in the sense of a syntax tree. Merging the leaves with the same propositional variables then leads to a cyclic graph. The motivation of using the ``syntax graph'' treewidth is that independent subformulas, \ie, subformulas without common variables, intuitively can be handled independently of each other. If many subformulas are connected by common variables this is reflected by a high treewidth.

\begin{define}[Formula treewidth and pathwidth]
For a formula $\varphi$ its treewidth $\tw(\varphi)$, resp., pathwidth $\pw(\varphi)$ is simply defined as $\tw(\calS_\varphi)$, resp., $\pw(\calS_\varphi)$.
\end{define}

The syntax graph as defined above is a generalisation of the \emph{incidence graph} of a CNF formula, and the incidence graph of a CNF is contained as a graph minor in its syntax graph: Simply merge all propositional variables with its negations, then for every clause contract all edges that belong to it. Then delete the disjunction nodes above the clauses. Therefore the structural treewidth is an upper bound for the incidence treewidth, the same holds for the pathwidth. This implies that all hardness results regarding the structural treewidth or pathwidth also hold for the incidence treewidth or pathwidth.

\section{Parameterised Complexity Results}\label{sec:results}

\subsection{Satisfiability}

\begin{theorem}\label{thm:ltl-hardness}
For $\F \in T$, $\G \in T$ or $\U \in T$, the problems \emph{$(\LTLSAT(T), \td + \pw_\varphi)$} and \emph{$(\LTLSAT(T), \td + \tw_\varphi)$} are $\W{1}$-hard, \ie, \emph{$\LTLSAT(T)$} parameterised by temporal depth together with syntactical pathwidth, resp., treewidth.
\end{theorem}
\begin{proof}
The result is proven by an fpt-reduction from the parameterised problem p-PW-SAT which was shown to be $\W{1}$-hard by Praveen \cite{prav13}.
An instance of p-PW-SAT is a tuple $I = \left(\varphi,\, k,\, (Q_i)_{i \in [k]},\, (C_i)_{i \in [k]}\right)$ where $\varphi$ is a propositional formula in CNF with variables $\{ q_1, \ldots, q_n \}$. The variables are partitioned into pairwise disjoint subsets $\{ Q_1, \ldots, Q_k\}$. The values $C_i \in \N$ are the \emph{capacities} of the partitions, \ie, the exact number of variables in $Q_i$ that must be set to true which is the \emph{weight} of the partition. An assignment is called \emph{saturated} if every partition has weight equal to its capacity.
For an instance $I$ we say that $I \in$ p-PW-SAT if $\varphi$ has an assignment that is both satisfying and saturated.
The parameter of p-PW-SAT is $\kappa(I) \dfn k + \pw(G_\varphi)$ where $G_\varphi$ is the primal
graph of
the CNF $\varphi$. The primal graph of a CNF is the graph that contains all propositions as vertices and edges for those variables that occur together in a clause.
For the reduction, we consider an LTL formula $\psi(I) \in \LTLform(\F,\G)$ that has constant temporal depth and $\kappa$-bounded pathwidth (and therefore treewidth). The formula $\psi(I)$ is a conjunction of the following subformulas:
\begin{alignat*}{2}
&\psi[\text{formula}] &&\,\,\,\dfn \varphi,
\\
&\psi[\text{depth}] &&\,\,\,\dfn \G \bigwedge_{i = 0}^{n - 1} \Big[(d_i \land \neg d_{i + 1})
  \rightarrow (m_{i \bmod 2} \\
& && \hphantom{\,\,\,\dfn \G \bigwedge_{i = 0}^{n - 1} \Big[} \land \neg m_{1- (i \bmod 2)} \land \F (d_{i + 1} \land \neg d_{i + 2})) \Big],
\\
&\psi[\text{fixed-}Q] &&\,\,\,\dfn \bigwedge_{i = 1}^{n} \left[\left(q_i \rightarrow \G q_i\right)
  \land \left(\neg q_i \rightarrow \G \neg q_i \right) \right],
\\
&\psi[\text{signal}] &&\,\,\,\dfn \G \bigwedge_{i = 1}^{n} \bigg[\left(d_i \land \neg d_{i + 1}\right) \rightarrow
\\
& && \qquad \bigg(\left(q_i \rightarrow \top^\uparrow_{p(i)}\right) \land \, \left(\neg q_i \rightarrow \bot^\uparrow_{p(i)}\right)\bigg)\bigg],\nonumber
\\
&\psi[\text{init}] \!\!\!&& \,\,\,\dfn d_0 \land \neg d_1
\land \G \bigwedge_{p = 1}^{k}\left[\top^{0}_{p} \land \bot^0_p \right],
\\
&\psi[\text{count}] &&\,\,\,\dfn \G \bigwedge_{p = 1}^{k} \bigwedge_{j = 0}^{\size{Q_p}} \bigwedge_{m = 0}^{1}
\bigg[ \\
& \qquad \qquad \bigg( \left(\top^\uparrow_{p} \land \top^j_{p} \land m_i\right) \rightarrow \G \left(m_{1-i} \rightarrow \G \top^{j+1}_{p} \right) \bigg) \span \span
\\
& \qquad \qquad \bigg( \left(\bot^\uparrow_{p} \land \bot^j_{p} \land m_i\right) \rightarrow \G \left(m_{1-i} \rightarrow \G \bot^{j+1}_{p} \right) \bigg) \bigg], \span \span
\\
&\psi[\text{monotone}]\!\!\! && \,\,\, \dfn \G \Bigg[\bigwedge_{i = 1}^{n} \left(d_i \rightarrow d_{i - 1} \right) \\
& \qquad \qquad \qquad \land \bigwedge_{p = 1}^{k} \bigwedge_{j = 1}^{\size{Q_p} + 1} \left(\top^j_{p} \rightarrow \top^{j-1}_{p} \right) \land \left(\bot^j_{p} \rightarrow \bot^{j-1}_{p} \right) \Bigg], \span \span
\\
&\psi[\text{target}] \!\!\!&& \,\,\, \dfn \G \bigwedge_{p=1}^{k} \bigg[\neg \top^{C_p + 1}  \land \neg \bot^{n(p) - C_p + 1}_p \bigg].
\end{alignat*}

The idea of the reduction is as follows. Let $\calM$ be a model of $\psi(I)$. For each proposition $q_i$ of $\{q_1, \ldots, q_n\}$ a world $w_i$ is contained in $\calM$ which has $d_i$ labeled. If $q_i$ is in partition $p(i)$, then in $w_i$ the ``signal'' proposition $\top^\uparrow_p(i)$, resp., $\bot^\uparrow_p(i)$ indicates that the weight counter of this partition should be increased, where for each partition the ones and zeros are counted separately. $\psi[\text{count}]$ implements the counting, whereas $\psi[\text{target}]$ ensures that every partition has exactly the desired weight (neither ones nor zeros are too many).

The exact proof of correctness and of the $\kappa$-boundedness of the pathwidth is omitted. The reader is instead referred to the thesis of \citeauthor{lueck-ma} \cite{lueck-ma}. Using the equivalences $\F \alpha \equiv \neg \G \neg \alpha$, $\G \alpha \equiv \neg \F \neg \alpha$ and $\F \alpha \equiv \top \U \alpha$, the reduction is possible with any temporal operator except $\X$.
\end{proof}

The next theorems were originally shown by \citeauthor{demri_complexity_2002}. \Cref{thm:slice-hardness} directly translates this into the world of parameterised complexity.

\begin{theorem}[\cite{demri_complexity_2002}]\label{thm:ltl-hardness-td-paraNP}
If $T \subseteq \{ \X \}$ or $T \subseteq \{ \F, \G \}$, then the problem \emph{$(\LTLSAT(T), \td)$} is \emph{para-$\NP$}-complete.
\end{theorem}
\begin{theorem}[\cite{demri_complexity_2002}]\label{thm:ltl-hardness-td-paraPSPACE}
If $\{ \F, \X \} \subseteq T$, $\{ \G, \X \} \subseteq T$ or $\U \in T$, then the problem \emph{$(\LTLSAT(T), \td)$} is \emph{para-$\PSPACE$}-complete.
\end{theorem}

\begin{theorem}\label{thm:ltl-x-in-fpt}
For $T \subseteq \{ \X \}$, the problems \emph{$(\LTLSAT(T), \tw_\varphi)$} and \emph{$(\LTLSAT(T), \pw_\varphi)$} are in \emph{$\FPT$}.
\end{theorem}
\begin{proof}
It holds due to the path semantics of LTL that $\X(\phi \land \psi) \equiv \X\phi \land \X \psi$,
$\X(\phi \lor \psi) \equiv \X\phi \lor \X \psi$ and $\X \neg \phi \equiv \neg \X \phi$
for $\phi, \psi \in \LTLform$.
Hence every LTL formula with only $\X$-operators can efficiently be
converted to an equivalent Boolean combination $\beta$ of $\X$-preceded variables:
\[
\varphi \equiv \beta(\X^{n_1} q_1, \ldots, \X^{n_m} q_m)\text{,} \quad \X^i \dfn \underbrace{\X \ldots \X}_{i \text{ times}}\text{,}
\]
where the $q_i$ are propositional variables. Inconsistent literals can only occur inside the same world and therefore at the same
nesting depth of $\X$. Hence the above formula $\varphi$ is satisfiable if and only if
it is satisfiable as a purely propositional formula where the expression $\X^{n_i} q_i$
is interpreted as an atomic formula (\ie, a variable).

Formally we have
$
(\LTLSAT(\X),\tw_\varphi) \leqfpt (\SAT,\tw_\varphi)\text{.}
$
Note that $(\SAT, \tw_\varphi) \in \FPT$ as a special case of CTL was shown by \citeauthor{pc-ctl15} \cite{pc-ctl15}.
As pathwidth is an upper bound for treewidth, $(\LTLSAT(\X), \pw_\varphi)$ is in $\FPT$ as well.
\end{proof}

Unsurprisingly, all hard LTL fragments correspond to hard CTL fragments if we just supplement the operators with path quantifiers \cite{pc-ctl15}.
For the fragment CTL($\AX$) (or equivalently modal logic on serial frames) being in $\FPT$ \cite{prav13} we however need the temporal depth as an additional parameter.
As satisfiable LTL formulas are already satisfied on paths and less expressive than modal logic, this extra parameter is not required for the $\LTLform(\X)$ fragment.

In the next part we turn from satisfiability to model checking.

\subsection{Model checking}

It turns out that LTL model checking is surprisingly hard for almost all studied parameterisations.
This is already the case in classical complexity theory. While model checking for CTL is $\P$-complete, it is $\PSPACE$-complete for LTL and \CTLstarText{}, and is in fact $\NP$-hard for every fragment with a non-empty operator set \cite{bamuscscscvo11}. This inherent hardness is due to the different semantics of CTL and LTL:
In CTL, every subformula of a formula is what is called \emph{state formula}. A polynomial time algorithm is obtained by recursively determining fulfilled subformulas in every world of the model.
LTL is built from \emph{path formulas} which have no truth value \wrt to states but only to paths from the root of the model.
This forbids P algorithms in the CTL style; hardness reductions to LTL model checking usually construct Kripke structures with few worlds and (exponentially) many paths between them. In fact, LTL model checking on non-branching structures is in P \cite{markey2003model}. The reductions in this section follow this scheme and in general produce branching Kripke structures with certain properties.

Long before the introduction of parameterised complexity theory, statements as early as from \citeauthor{lichPnue85} already distinguished between \emph{program complexity}, the runtime dependent on the length of the formula $\varphi$, and \emph{structure complexity}, the runtime dependent on the length of the structure $\calA$ to be checked \cite{lichPnue85}. They stated that the runtime factor $2^\size{\varphi}$ in their algorithm does not prohibit efficient model checking as the size of the structure is clearly dominant in practice.


In the context of parameterised complexity, this automatically yields nice fixed-parameter tractable problems:

\begin{corollary}
Let $\kappa(\varphi, \calA, w) \dfn \size{\varphi}$. Then \emph{$(\LTLMC, \kappa) \in \FPT$}.
\end{corollary}

The next logical step is to study the influence of more parameterisations on the model checking complexity:
Define $\tw_\calA$ as the treewidth of the input structure, \ie, $\tw_\calA(\varphi, \calA, w) \dfn \tw(\calA)$.
Define $\tw_\varphi$ as the structural treewidth of the input formula, \ie, $\tw_\varphi(\varphi, \calA, w) \dfn \tw(\calS_{\varphi})$. Similarly define $\pw_\calA$ and $\pw_\varphi$.


\begin{proposition}[\cite{SC85, sch02, demri_complexity_2002}]\label{thm:ltlmc-in-pspace}
\emph{$\LTLMC(\X, \F, \G, \U) \in \PSPACE$}.
\end{proposition}

\begin{proposition}[\cite{SC85, sch02, demri_complexity_2002}]\label{thm:ltlmc-in-conp}
For $T \subseteq \{ \X \}$ or $T \subseteq \{ \F, \G \}$ it holds \emph{$\LTLMC(T) \in \coNP$}.
\end{proposition}

\begin{define}[Maximum branching degree]
Let $\calA$ be a Kripke structure. Then write $\Delta(\calA)$ for the maximum branching degree in $\calA$, \ie, the smallest number $\Delta$ \st every world in $\calA$ has at most $\Delta$ successors.
\end{define}

\begin{theorem}
\label{thm:ltlmc-f-conp}\emph{$(\LTLMC(\F), \td + \Delta)$} is complete for \emph{para-$\coNP$}.
\end{theorem}
\begin{proof}
We follow \citeauthor{SC85} who showed that $\LTLMClim{1}(\F)$ (\ie, only $\F$-operators without nesting) is $\coNP$-hard.
This is done by a reduction from the complement of the $\NP$-complete 3SAT problem: Given a propositional formula $\varphi$ in 3CNF, is it satisfiable?
 
For this we construct a formula $\psi \in \LTLlim{1}(\F)$ and a structure $\calS$ with constant branching such that $(\calS,w_0) \models \psi$ if and only if $\varphi$ is unsatisfiable.
First assume $\varphi = \bigwedge_{i = 1}^{m} \left(L_{i,1} \lor L_{i,2} \lor L_{i,3}\right)$ where $L_{i,j}$ is a literal, \ie, a propositional variable or its negation. Then simply define $\psi \dfn \bigvee_{i = 1}^{m} \left(\F \neg L_{i,1} \land \F \neg L_{i,2} \land \F \neg L_{i,3}\right)$, so $\psi$ is basically the negation of $\varphi$ supplemented with $\F$ operators in front of the literals.

Assume that $\varphi$ contains variables $\{x_1, \ldots, x_n\}$.
For a correct reduction the structure $\calS$ is now required to allow either $\F x_i$ or $\F \neg x_i$ to hold for $1 \leq i \leq n$, but not both. Also, for every subset $X \subseteq \{ x_1, \ldots, x_n \}$ of variables (which can
be interpreted as the assignment that sets exactly the variables in $X$ to true) there should be a path through $\calS$ and vice versa.
The structure depicted in \Cref{fig:ltl-mc-conp} has these property and therefore models propositional assignments as runs from its initial world $w_0$. This means that all runs in $\calS$ fulfill the path formula $\psi$ if and only if $\neg \varphi$ is satisfied by all Boolean assignments. Hence $\varphi \notin \text{3SAT} \Leftrightarrow (\psi, \calS, w_0) \in \LTLMClim{1}(\F)$. $\psi$ and $\calS$ are both constructible in linear time.
The para-$\coNP$-completeness follows from \Cref{thm:slice-hardness} and \Cref{thm:ltlmc-in-conp}.
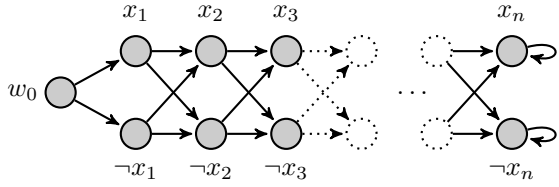
\begin{figure}\centering
\begin{tikzpicture}[->,>=stealth',shorten >=1pt,auto,node distance=0.55cm,%
        thick, %
        world/.style={circle,fill=black!20,draw,minimum size=0.4cm,inner sep=0pt}, %
        world2/.style={circle,dotted,draw,minimum size=0.4cm,inner sep=0pt}, %
        dummy/.style={inner sep=0pt,node distance=1cm}]

\node[world] (s0) at (0,0) {};
\node[dummy] (s1) [right of=s0] {};
\node[dummy] (s2) [right of=s1] {};
\node[dummy] (s3) [right of=s2] {};
\node[dummy] (s4) [right of=s3] {};
\node[dummy] (sll) [right of=s4] {};
\node[dummy] (sl) [right of=sll] {};
\node[dummy] (dots) [node distance=0.5cm, right of=s4] {$\quad \ldots$};

\node[world] (s1u) [above of=s1] {};
\node[world] (s2u) [above of=s2] {};
\node[world] (s3u) [above of=s3] {};
\node[world2] (s4u) [above of=s4] {};
\node[world2] (sllu) [above of=sll] {};
\node[world] (slu) [above of=sl] {};

\node[world] (s1d) [below of=s1] {};
\node[world] (s2d) [below of=s2] {};
\node[world] (s3d) [below of=s3] {};
\node[world2] (s4d) [below of=s4] {};
\node[world2] (slld) [below of=sll] {};
\node[world] (sld) [below of=sl] {};

\node[dummy] (l0) [node distance=0.5cm, left of=s0] {$w_0$};
\node[dummy] (x1) [node distance=0.5cm, above of=s1u] {$x_1$};
\node[dummy] (x2) [node distance=0.5cm, above of=s2u] {$x_2$};
\node[dummy] (x3) [node distance=0.5cm, above of=s3u] {$x_3$};
\node[dummy] (x1n) [node distance=0.5cm, below of=s1d] {$\neg x_1$};
\node[dummy] (x2n) [node distance=0.5cm, below of=s2d] {$\neg x_2$};
\node[dummy] (x3n) [node distance=0.5cm, below of=s3d] {$\neg x_3$};
\node[dummy] (xl) [node distance=0.5cm, above of=slu] {$x_n$};
\node[dummy] (xln) [node distance=0.5cm, below of=sld] {$\neg x_{n}$};

\path[->]
(s0) edge node{} (s1u)
(s0) edge node{} (s1d)
(s1u) edge node{} (s2u)
(s1u) edge node{} (s2d)
(s1d) edge node{} (s2u)
(s1d) edge node{} (s2d)
(s2u) edge node{} (s3u)
(s2u) edge node{} (s3d)
(s2d) edge node{} (s3u)
(s2d) edge node{} (s3d)
(sllu) edge node{} (slu)
(sllu) edge node{} (sld)
(slld) edge node{} (slu)
(slld) edge node{} (sld)
(slu) edge [loop right] node {} (slu)
(sld) edge [loop right] node {} (sld);

\path[dotted,->]
(s3u) edge node{} (s4u)
(s3u) edge node{} (s4d)
(s3d) edge node{} (s4u)
(s3d) edge node{} (s4d);
\end{tikzpicture}
\caption{Structure that models propositional assignment as runs.}
\label{fig:ltl-mc-conp}
\end{figure}
\end{proof}


This results may be surprising at first sight. As LTL is easy on non-branching structures, one could expect that bounding the branching degree leads to an easy problem as well, but in fact already a branching of degree two is sufficient to express $\coNP$-hard model properties.

\begin{theorem}\label{thm:ltl-x-cowp}
Let $T \subseteq \{ \X \}$. Then \emph{$(\LTLMC(T), \td) \in \coW{\P}$}.
\end{theorem}
\begin{proof}
Let a formula $\varphi \in \LTLform(\X)$ and a structure $\calA$ with root $w$ be given. Now it holds that $(\calA,w) \not\models \varphi$ if and only if there is a path $\pi \in \Pi(w)$ \st $(\calA, \pi) \not\models \varphi$. But this depends only on a finite prefix of $\pi$ that has length $\td(\varphi)$.
So to determine if the given formula is not satisfied by the structure, simply guess a finite path of length $\td(\varphi)$ through $\calA$ and verify the formula. This requires at most $\bigO{\td(\varphi) \cdot \log \size{\calA}}$ non-deterministic steps (using pointers to denote worlds) and leads to $\coW{\P}$ according to \Cref{def:wp}.
\end{proof}

\begin{theorem}\emph{$(\LTLMC(T), \td + \Delta) \in \FPT$} for $T \subseteq \{ \X \}$.\end{theorem}
\begin{proof}
Instead of guessing a path of length $\td(\varphi)$ for given $\varphi$ as in \Cref{thm:ltl-x-cowp} do an exhaustive search by iterating the successors of each world. This leads to a runtime of $\size{\varphi}^\bigO{1} \cdot \Delta(\calA)^{\td(\varphi)}$.
\end{proof}

Note that this tractability result differs from the intractability result of \Cref{thm:ltlmc-f-conp}
only in the allowed temporal operator. This shows that the expressive power of $\X$ is (obviously) tightly bounded to its allowed nesting depth.

\medskip


\begin{theorem}\label{thm:ltlmc-paranp-hard}For any non-empty set $T$ of temporal operators \emph{$(\LTLMC(T), \pw_\calA + \Delta)$} and \emph{$(\LTLMC(T), \tw_\calA + \Delta)$} are hard for \emph{para-$\coNP$}.
\end{theorem}
\begin{proof}
This is shown using the same argument as in \Cref{thm:ltlmc-f-conp}. 
The structure $\calS$ used there has a constant pathwidth (and therefore treewidth) of at most three:
One bag $B_0$ contains the worlds $w_0, w_1^+, w_1^-$, where $w_i^+$ ($w_i^-$) is the unique world in $\calS$ that has $x_i$ ($\neg x_i$) labeled. Further bags $B_i$ contain $w_i^+, w_{i+1}^+, w_i^-$ and $w_{i+1}^-$ for $1 \leq i < n$. Connecting the bags $B_i$ and $B_{i+1}$ for $0 \leq i < n$ results in a path-decomposition of $\calS$ with width at most three. The maximum branching of the structure is as well constant. For $\X \in T$ we modify the reduction given in the proof of \Cref{thm:ltlmc-f-conp}. Simply replace the subformula $\F L_i$ by $\X^i L_i$ in the formula and the reduction stays valid. This substitution leads only to polynomial blowup.
\end{proof}
\begin{corollary}
If $\F \in T$ or $\G \in T$, then \emph{$(\LTLMC(T), \pw_\calA + \Delta + \td)$} and \emph{$(\LTLMC(T), \tw_\calA + \Delta + \td)$} are hard for \emph{para-$\coNP$}.
\end{corollary}

\begin{theorem}
If $\U \in T$ or $\{\X,\F\} \subseteq T$ or $\{\X,\G\} \subseteq T$, then \emph{$(\LTLMC(T), \pw_\calA + \Delta + \td)$} and \emph{$(\LTLMC(T), \tw_\calA + \Delta + \td)$} are complete for \emph{para-$\PSPACE$}.
\end{theorem}
\begin{proof}
The hardness is shown similar to the reductions in Section 5 and 6 of \cite{demri_complexity_2002}. To obtain
a constant out-degree, resp., pathwidth in the structure, the problem 3QBF instead of QBF must be used which is also $\PSPACE$-complete.
3QBF is defined like QBF but has a matrix in 3CNF. The final world shown in Figure 5 of \cite{demri_complexity_2002}
has still a larger out-degree, but it can be cloned to a binary tree with out-degree two. This does not influence the satisfaction of the $\U$ formula as $\U$ is stutter-invariant.
\end{proof}

The presented results imply that even very simple Kripke structures are hard to check for certain properties. This is consistent to usual LTL model checking algorithms since they already have runtime exponential in $\size{\varphi}$ but only linear in $\size{\calA}$. Therefore we now turn to more parameterisations of the LTL formula itself.

\medskip

Let denote with $\nVar(\varphi)$ the number of distinct propositional variables occurring in $\varphi$.

\begin{theorem}\label{thm:ltl-one-variable-hardness}
The parameterised problems \emph{$(\LTLMC(\X), \tw_\varphi)$}, \emph{$(\LTLMC(\X), \nVar)$}, \emph{$(\LTLMC(\F), \tw_\varphi)$} and \emph{$(\LTLMC(\F), \nVar)$} are complete for \emph{para-$\coNP$}.
\end{theorem}
\begin{proof}
The hardness is shown by \citeauthor{demri_complexity_2002} by reduction to a formula with constant number of propositions.
This also leads to a formula with a constant structural treewidth as every syntax graph without variables already is a tree (\ie, has treewidth one) and every propoition increases the treewidth only at most one.
\end{proof}
%


\begin{theorem}\label{thm:ltl-f-tw-even-odd-conph}
\emph{$(\LTLMC(\F), \tw_\varphi)$} and \emph{$(\LTLMC(\F), \nVar)$} are complete for \emph{para-$\coNP$}.
\emph{$(\LTLMC(\F, \X), \tw_\varphi)$} and \emph{$(\LTLMC(\F, \X), \nVar)$} are complete for \emph{para-$\PSPACE$}.
\end{theorem}
\begin{proof}
Shown in \cite{demri_complexity_2002}. $\X,\F$ can encode the PSPACE-hard QBF problem in a model checking instance using only a single proposition, but label this proposition in worlds of different depth. If the only operator is $\F$, then propositional satisfiability can be simulated. However then two variables are required, as a world of certain depth can only be exactly reached by $\F$ using a certain form of alternation. This ``alternation'' trick is necessary since future is reflexive (like in most temporal logics) and hence the present is a part of the future.
\end{proof}

It seems that in the reductions a large temporal depth can neither be avoided for $\F$ nor for $\X$. Therefore it is unlikely that para-$\coNP$-hardness or para-$\PSPACE$-hardness for model checking stays \wrt the parameterisation $\td + \tw_\varphi$. Therefore we now consider variants of the \emph{Tiling} problem for alternative reductions.

\begin{define}[Tiling]
Let $C$ be a finite set of \emph{colors} and $D \subseteq C^4$ a set of \emph{tiles}. Every tile has four sides, namely \emph{up}, \emph{down}, \emph{left} and \emph{right}, which each have a color $c \in C$.
Use the quadruple notation to explicitly write the colors of a tile: $d = \langle c_u, c_d, c_l, c_r \rangle$.
A \emph{$D$-tiling} for a discrete area $R \subseteq \N \times \N$ is a function $\gamma : R \to D$.

Write $\gamma(x,y) = \langle (x,y)_u, (x,y)_d, (x,y)_l, (x,y)_r \rangle$.
Then $\gamma$ is a \emph{valid tiling} if for every $(x,y), (x',y') \in R$ holds:
\begin{itemize}
  \item if $x' = x$ and $y' = y + 1$, then $(x,y)_d = (x',y')_u$,
  \item if $x' = x + 1$ and $y' = y$, then $(x,y)_r = (x',y')_l$.
\end{itemize}
\end{define}

\begin{define}
The problem \textsc{SquareTiling} contains the instances $(C, D, \langle k \rangle_1)$ for which the $k\times k$-grid has a valid $D$-tiling. Here, $\langle \cdot \rangle_1$ denotes a unary encoding.
\end{define}


\begin{proposition}[\cite{tiling}]
Let $\kappa(C, D, \langle k \rangle_1) \dfn k$. Then the parameterised problem $(\text{\textsc{SquareTiling}},\kappa)$
is \emph{$\W{1}$}-complete.
\end{proposition}

\begin{theorem}\label{thm:ltl-mc-is-w1-hard}
\emph{$(\LTLMC(\X), \td + \pw_\varphi)$} and \emph{$(\LTLMC(\X), \td + \tw_\varphi)$} are \emph{$\coW{1}$}-hard.
\end{theorem}
\begin{proof}
The idea of the proof, a reduction from \textsc{SquareTiling} to the complement of the problems in the theorem, is to use the path semantics of LTL to describe a valid tiling of the $k \times k$ grid: For every \textsc{SquareTiling} instance $(C, D, \langle k \rangle_1)$ we construct a formula $\psi$ and structure $\calS$. $\psi$ will have $k$-bounded temporal depth and structural pathwidth. The Kripke structure $\calS$ however will have unbounded branching degree $\Delta$ (which is unlikely to be avoided as $\LTLMC(\X)$ is already in $\FPT$ with parameter $\td + \Delta$).

Construct $\calS$ as follows:
\begin{itemize}
  \item Add worlds $w_\text{start}$ and $w_\text{end}$ with the proposition
  $q_\text{end}$ labeled in $w_\text{end}$.
  \item For every tile $d \in D$ and for every $i \in [k^2]$ add a world $w^i_{d}$.
  \item Connect $w_\text{start}$ to $w^1_{d}$ for every $d \in D$.
  \item Connect $w^{\,k^2}_{d}$ to $w_\text{end}$ for every $d \in D$.
  \item Connect $w^i_{d}$ to $w^{i+1}_{d'}$ for every $d, d' \in D$ and $1 \leq i < k^2$.
  \item Connect $w_\text{end}$ to itself.
  \item In every world $w^i_{d}$ with $d = \langle c_u, c_d, c_l, c_r \rangle$
  label propositional variables $c^i_u$, $c^i_d$, $c^i_l$, $c^i_r$.
  \item In every world $w^i_{d}$ where $i = k \cdot j$ for $j \in [k]$ label a propositional
  variable $q_\text{border}$.
\end{itemize}

\begin{figure}
\centering
\begin{tikzpicture}[->,>=stealth',shorten >=1pt,auto,node distance=0.8cm,%
        thick, %
        world/.style={circle,fill=black!20,draw,minimum size=0.4cm,inner sep=0pt}, %
        world2/.style={circle,dotted,draw,minimum size=0.4cm,inner sep=0pt}, %
        dummy/.style={inner sep=0pt,node distance=0.8cm},yscale=0.75]

\node[world] (s0) at (0cm,-3.6cm) {};
\node[world] (sEnd) at (6cm,-3.6cm) {};
\node[dummy] (l0) [node distance=0.49cm, above left of=s0] {$w_\text{start}$};
\node[dummy] (l0) [node distance=0.45cm, above right of=sEnd] {$w_\text{end}$};

\node[dummy] (dots) at (3.5cm,-2.0cm) {$\ldots$};
\node[dummy] (dots2) at (3.5cm,-4.5cm) {$\ldots$};
\node[dummy] (dots3) at (1.5cm,-3.4cm) {$\vdots$};
\node[dummy] (dots4) at (4.5cm,-3.4cm) {$\vdots$};

\foreach \i in {1,...,2}{
\foreach \d in {1,...,2}{
\node[world] (s\i\d) at (1cm * \i, -1cm * \d) {};
}}
\foreach \i in {1,...,2}{
\node[world] (s\i5) at (1cm * \i, -1cm * 5) {};
}
\foreach \i in {1,...,5}{
\node[world2] (s\i3) at (1cm * \i, -1cm * 3) {};
\node[world2] (s\i4) at (1cm * \i, -1cm * 4) {};
}
\foreach \d in {1,...,2}{
\node[world2] (s3\d) at (1cm * 3, -1cm * \d) {};
\node[world2] (s4\d) at (1cm * 4, -1cm * \d) {};
\node[world] (s5\d) at (1cm * 5, -1cm * \d) {};
}

\node[world2] (s35) at (1cm * 3, -1cm * 5) {};
\node[world2] (s45) at (1cm * 4, -1cm * 5) {};
\node[world] (s55) at (1cm * 5, -1cm * 5) {};

\node[dummy] (l51) at (1cm * 5, -1cm * 1 + 0.7cm) {$w^{\,k^2}_{d_1}$};
\node[dummy] (l55) at (1cm * 5, -1cm * 5 - 0.7cm) {$w^{\,k^2}_{d_n}$};
\foreach \i in {1,...,2}{
\node[dummy] (l\i1) at (1cm * \i, -1cm + 0.7cm) {$w^{\,\i}_{d_1}$};
\node[dummy] (l\i5) at (1cm * \i, -1cm * 5 - 0.7cm) {$w^{\,\i}_{d_n}$};
}

\foreach \d in {1,...,2}{
\foreach \e in {1,...,2}{
\path[->] (s1\d) edge node{} (s2\e);
}}

\path[->,dotted]
(s0) edge node{} (s14)
(s0) edge node{} (s13);
\path[->] (s0) edge node{} (s15);

\foreach \i in {1,...,1}{
\foreach \d in {1,...,3}{
\pgfmathparse{int(\i + 1)}
\path[->,dotted] (s\i\d) edge node{} (s\pgfmathresult3);
}}

\foreach \d in {1,...,3}{
\foreach \e in {1,...,3}{
\path[->,dotted] (s2\d) edge node{} (s3\e);
\path[->,dotted] (s4\d) edge node{} (s5\e);
}}

\foreach \d in {4,...,5}{
\foreach \e in {4,...,5}{
\path[->,dotted] (s4\d) edge node{} (s5\e);
}}

\foreach \i in {1,...,2}{
\foreach \d in {4,...,5}{
\foreach \e in {4,...,5}{
\pgfmathparse{int(\i + 1)}
\path[->,dotted] (s\i\d) edge node{} (s\pgfmathresult\e);
}}}

\foreach \d in {1,...,2}{
\path[->] (s0) edge node{} (s1\d);
\path[->] (s5\d) edge node{} (sEnd);
}

\path[->]
(s15) edge node{} (s25)
(s55) edge node{} (sEnd)
(sEnd) edge [loop right] node{} (sEnd);

\path[->,dotted]
(s53) edge node{} (sEnd)
(s54) edge node{} (sEnd);

\path[->,dotted] (s13) edge node{} (s21)
(s13) edge node{} (s22)
(s13) edge node{} (s23);
\end{tikzpicture}
\caption{Structure that models square tilings as runs.}\label{fig:ltlmc-w1-tiling}
\end{figure}
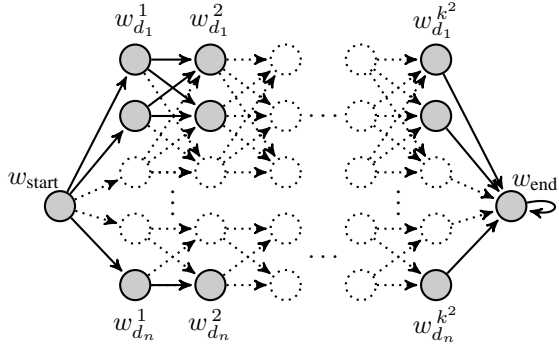

The structure $\calS$ is illustrated in \Cref{fig:ltlmc-w1-tiling}.
It models (not necessarily valid) tilings $\gamma$ as runs from $w_\text{start}$ by ``serializing'' the square into a path: It contains $k$ worlds for the first row, another $k$ worlds for the second row appended to the first $w$ worlds, and so on to the $k$-th row, resulting in a total of $k^2$ worlds on each path (besides $w_\text{start}$ and $w_\text{end}$). At the same time, there are $\size{D}$ successors that a path can use to select the next tile in the current (or next) row: For every place $(i,j) \in [k] \times [k]$ a tile $d$ is selected by visiting the corresponding world $w^{(i-1)\cdot k + j}_{d}$.

Now we give the path formulas that verify that the tiling $\gamma$ described by a run $\pi \in \Pi(w_\text{start})$ is valid.
\begin{align*}
\psi^i_{c,r} &\dfn \left[q_\text{border} \lor %
\left( c^i_r \rightarrow \X( q_\text{end} \lor c^{i+1}_l )\right)\right]\\
\psi^i_{c,d} &\dfn \left[c^i_d \rightarrow \X^k\left(q_\text{end} \lor c^{i+k}_u\right) \right]
\end{align*}

The formula $\psi^i_{c,r}$ is true in a path $\pi$ starting in a world $w^i_{d}$ (which has color $c$ to the right) if $\pi$ chooses a matching successor: Either $w^i_d$ is a \emph{border} and the color to the right is not relevant, or $w^i_d$ has $w_\text{end}$ as immediate successor and the tiling is complete, or the color matches the \emph{left} color of the next tile.

Similar, the formula $\psi^i_{c,d}$ ensures that the tile directly below the current one (which lies in distance exactly $k$ in the structure) has the matching \emph{up} color or is already beyond the last row.

We form the conjunction over every color $c$ and every $i$:
\[
\psi \dfn \bigwedge_{i = 1}^{k^2} \left[ \X^i \bigwedge_{c \in C} \left(\psi^i_{c,r} \land \psi^i_{c,d} \right) \right]
\]

\begin{claim}
$\calS$ and $\psi$ can be constructed in fpt-time.
\end{claim}
\begin{proof2}
The structure $\calS$ can be constructed in time $\bigO{\size{D}^2 \cdot k^2}$ and the formula $\psi$ can be constructed in time $\bigO{\size{C} \cdot k^3}$, which is both polynomial since $k$ is encoded unary in \textsc{SquareTiling}. \qedBlack
\end{proof2}

\begin{claim}
Let $\pi = (w_\text{start}, w^1_{f(1)}, w^2_{f(2)}, \ldots, w^{k^2}_{f(k^2)}, w_\text{end}, \ldots)$ be a run through $\calS$,
where $f \colon [k^2] \to D$ selects the tile at each step.
Then $\pi \models \psi$ if and only if $f(1), f(2), \ldots, f(k^2)$ form a valid tiling of $[k] \times [k]$.
\end{claim}
\begin{proof2}
``$\Rightarrow$'': Assume there are $(x,y)$ and $(x',y')$ such that the tiling conditions are violated.
\begin{itemize}
  \item Case 1: $x' = x + 1$ and $y' = y$. Then $x \neq k$, $(x,y)$ has right color $c$ and $(x+1,y)$ has left color $c' \neq c$. Let $i \dfn (x - 1) \cdot k + y$. By definition of $\pi$ it holds that $\pi[1] = w_\text{start}$ and $\pi[i + 1] = w^i_{f(i)}$. But as $w^i_{f(i)}$ is not a border and $\pi_{\geq i + 1} \models \psi^i_{c,r}$, so the successor has $c$ as left color. But this means that $c^i_r$ is labeled in $w^i_{f(i)}$ and $c^{i+1}_l$ is labeled in $w^{i+1}_{f(i+1)}$, contradiction to $c' \neq c$.
  \item Case 2: $x' = x$ and $y' = y + 1$ which is similar proven as Case 1.
\end{itemize}

``$\Leftarrow$'': Let $f(1), f(2), \ldots, f(k^2)$ be a valid tiling $\gamma$ of $[k] \times [k]$. Assume that $\neg \psi$ holds, \ie, there is a color $c$ and an $i$ \st{} $\pi_{\geq i + 1}$ does not satisfy $\psi^i_{c,r} \land \psi^i_{c,d}$. If $\psi^i_{c,r}$ is false, then $w^i_{f(i)}$ is not a border but also has a different \emph{right} color than its successor on $\pi$ has as \emph{left} color. But then $\gamma$ would not be a valid tiling. The case that $\psi^i_{c,d}$ is false can be handled analogously.
\qedBlack
\end{proof2}

It is easy to see that every run $\pi$ through $\calS$ from $w_\text{start}$ has the form as in the above claim, \ie, $\pi = (w_\text{start}, w^1_{f(1)}, w^2_{f(2)}, \ldots, w^{k^2}_{f(k^2)}, w_\text{end}, \ldots)$. Hence we get $(C, D, k) \in \text{\textsc{SquareTiling}}$ if and only if $\exists \pi \in \Pi(w_\text{start}) \, : \, (\calS,\pi) \models \psi$ if and only if $(\neg \psi, \calS, w_\text{start}) \notin \LTLMC(\X)$.

\begin{claim}
The formula $\psi$ has temporal depth $k^2 + k$ and structural pathwidth at most $2k^2 + k + 15$.
\end{claim}
\begin{proof2}
The temporal depth of $k^2 + k$ is the nesting depth of $\X$ operators in $\psi$.

For the pathwidth we construct a path-decomposition $\calP$ of $\psi$ as follows: For every $i \in [k^2]$ and every color $c \in C$ we create an isolated bag $B^i_c$. The bag $B^i_c$ contains the nodes representing
\begin{itemize}
  \item the Boolean connectives $\lor, \rightarrow$ and $\lor$ in $\psi^i_{c,r}$,
  \item the Boolean connectives $\rightarrow$ and $\lor$ in $\psi^i_{c,d}$,
  \item the variables $q_\text{border}$, $q_\text{end}$, $c^i_r$, $c^{i+1}_l$, $c^i_d$ and $c^{i+k}_u$,
  \item the single $\X$-operator in $\psi^i_{c,r}$,
  \item the $k$ $\X$-operators in $\psi^i_{c,d}$.
\end{itemize}

The isolated bag covers every edge between nodes representing subformulas of $\psi^i_{c,r}$ and $\psi^i_{c,d}$ with a width of $\size{B^i_c} = 3 + 2 + 6 + 1 + k = k + 12$. Also every subformula of $\psi^i_{c,r}$ and $\psi^i_{c,d}$
except $q_\text{border}$ and $q_\text{end}$ occurs exactly once in $\psi$, hence every such subformula of $\psi$ 
trivially induces a connected path in $\calP$. But as $q_\text{border}$ and $q_\text{end}$ are added into every bag $B^i_c$ they
also induce a connected path as soon as the bags are connected.

To handle the remaining connectives including the ``big conjunctions'' of size $\size{C}$, proceed as follows: First for every formula $\xi^i_c \dfn \left(\psi^i_{c,r} \land \psi^i_{c,d} \right)$, add $\xi^i_c$ to $B^i_c$.

Assume that the colors are ordered as $c_1, c_2, \ldots, c_\size{C}$, and that the big conjunctions have the structure
$((((\xi_1 \land \xi_2) \land \xi_3) \ldots) \land \xi_\size{C})$.
For every $j$, $1 \leq j < \size{C}$ then connect the bags $B^i_{c_j}$ and $B^i_{c_{j+1}}$ by inserting an edge in $\calP$, and add the $j$-th $\land$-node into both bags. Then after inserting the last conjunction, add the $i$ nodes for $\X^i$ to $B^i_{c_{\size{C}}}$, and finally add the nodes for the conjunction of size $k^2$ to every bag. These steps increase the size of every bag by at most $2 k^2 + 3$.

As $\calP$ now consists of $k^2$ disconnected sequences of $\size{C}$ bags each, concatenate them into a path in arbitrary order. This leads to $\calP$ being a connected path; and the variables $q_\text{border}$ and $q_\text{end}$ as well as the nodes of the $k^2$-conjunction now induce connected subpaths. \qedBlack
\end{proof2}
From the claims a valid fpt-reduction follows and thereby we conclude the theorem.
\end{proof}

Note that with only $\X$ the temporal depth of $k^2$ again is unlikely to be avoided, as for fixed temporal depth
the problem $\LTLMC(\X)$ is already solvable in logspace and in fact equivalent to checking assignments of 
propositional formulas \cite{demri_complexity_2002}.

\begin{theorem}\label{thm:ltl-mc-u-w1-hard}Let $\F \in T$, $\G \in T$ or $\U \in T$.
Then \emph{$(\LTLMC(T), \td + \pw_\varphi)$} and \emph{$(\LTLMC(T), \td + \tw_\varphi)$} are \emph{$\coW{1}$}-hard.
\end{theorem}
\begin{proof}
We adapt the reduction given in the proof of \Cref{thm:ltl-mc-is-w1-hard}.
First label a new depth proposition $d_i$ in every world $w^i_{d}$ for $d \in D$, $1 \leq i \leq k^2$. Then change the formulas as follows:
\begin{align*}
\psi^i_{c,r} &\dfn \left[q_\text{border} \lor %
\left(c^i_r \rightarrow \F( c^{i+1}_l )\right)\right]\\
\psi^i_{c,d} &\dfn \begin{cases}\left[c^i_d \rightarrow \F\left(c^{i+k}_u\right) \right] \text{ if $i + k \leq k^2$}\\\top \qquad \qquad \qquad \, \text{ otherwise}
\end{cases}\\
\psi &\dfn \bigwedge_{i = 1}^{k^2 - 1} \left[ \F \bigg(d_i \land \bigwedge_{c \in C} \left(\psi^i_{c,r} \land \psi^i_{c,d} \right) \bigg) \right]
\end{align*}

This does increase the pathwidth at most by $k^2$ for $d_1, \ldots, d_{k^2}$.
Also $\F\alpha$ can be replaced by $\neg \G \neg \alpha$ and $\top \U \alpha$ for the remaining cases.
\end{proof}
%


\begin{define}
The problem \textsc{RectangleTiling} contains the tuples $(C, c_0, c_1, D)$ for which there is an $m \in \N$ such that the $\size{D} \times m$-grid has a valid $D$-tiling $\gamma$ with the color $c_0$ at the top edge and $c_1$ at the bottom edge.


\end{define}

\begin{proposition}[\cite{tilingClassical}]
The problem \textsc{RectangleTiling} is \emph{$\PSPACE$}-complete.
\end{proposition}

\begin{theorem}\label{thm:ltl-xf-parapspace-h}
If $\{\X,\F\} \subseteq T$ or $\{\X,\G\} \subseteq T$, then the problems \emph{$(\LTLMC(T), \pw_\varphi)$} and \emph{$(\LTLMC(T), \tw_\varphi)$} are \emph{para-$\PSPACE$}-complete.
\end{theorem}
\begin{proof}
We consider a $\leqpm$-reduction from \textsc{RectangleTiling} to $\LTLMC(\X, \F)$ such that only LTL formulas with constant pathwidth (and therefore constant treewidth) are produced. This proves the theorem according to \Cref{thm:slice-hardness} and \Cref{thm:ltlmc-in-pspace}.
This reduction originally from \citeauthor{SC85} is to show the $\PSPACE$-hardness of general LTL model checking. We modify it to obtain a constant pathwidth.

Write the shortcut $n \dfn \size{D}$.
Then similar to the proof of \Cref{thm:ltl-mc-is-w1-hard} we construct a Kripke structure $\calS$ that models $n \times m$-tilings as runs:
\begin{itemize}
  \item Add worlds $w_\text{left}$, $w_\text{right}$ and $w_\text{end}$
  which each have only one proposition labeled, namely $q_\text{left}$, $q_\text{right}$, and $q_\text{end}$.
  \item For every tile $d \in D$ and every $i \in [n]$ add a world $w^i_{d}$.
  \item Connect $w_\text{left}$ to $w^1_{d}$ for every $d \in D$.
  \item Connect $w^{k}_{d}$ to $w_\text{right}$ for every $d \in D$.
  \item Connect $w^i_{d}$ to $w^{i+1}_{d'}$ for every $d, d' \in D$ and $1 \leq i < n$.
  \item Connect $w_\text{right}$ to $w_\text{end}$, $w_\text{right}$ to $w_\text{left}$ and $w_\text{end}$ to itself.
  \item In every world $w^i_{d}$ with $d = \langle c_u, c_d, c_l, c_r \rangle$
  label propositional variables $c_u$, $c_d$, $c_l$, $c_r$.
\end{itemize}

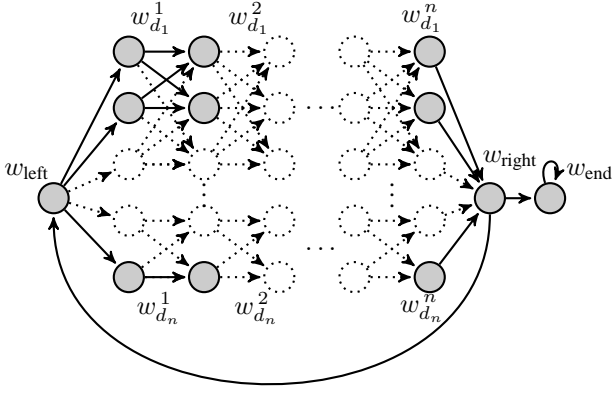
\begin{figure}
\centering
\begin{tikzpicture}[->,>=stealth',shorten >=1pt,auto,node distance=0.8cm,%
        thick, %
        world/.style={circle,fill=black!20,draw,minimum size=0.4cm,inner sep=0pt}, %
        world2/.style={circle,dotted,draw,minimum size=0.4cm,inner sep=0pt}, %
        dummy/.style={inner sep=0pt,node distance=0.8cm},yscale=0.75]

\node[world] (s0) at (0cm,-3.6cm) {};
\node[world] (sEnd) at (5.8cm,-3.6cm) {};
\node[world] (sEndEnd) at (6.6cm,-3.6cm) {};
\node[dummy] (l0) [node distance=0.5cm, above left of=s0] {$w_\text{left}$};
\node[dummy] (ll0) [node distance=0.5cm, above of=sEnd] {$\;\;\quad w_\text{right}$};
\node[dummy] (lll0) [node distance=0.5cm, above right of=sEndEnd] {$\quad w_\text{end}$};

\node[dummy] (dots) at (3.55cm,-2.0cm) {$\ldots$};
\node[dummy] (dots2) at (3.55cm,-4.5cm) {$\ldots$};
\node[dummy] (dots3) at (2cm,-3.4cm) {$\vdots$};
\node[dummy] (dots4) at (4.5cm,-3.4cm) {$\vdots$};

\foreach \i in {1,...,2}{
\foreach \d in {1,...,2}{
\node[world] (s\i\d) at (1cm * \i, -1cm * \d) {};
}}
\foreach \i in {1,...,2}{
\node[world] (s\i5) at (1cm * \i, -1cm * 5) {};
}
\foreach \i in {1,...,5}{
\node[world2] (s\i3) at (1cm * \i, -1cm * 3) {};
\node[world2] (s\i4) at (1cm * \i, -1cm * 4) {};
}
\foreach \d in {1,...,2}{
\node[world2] (s3\d) at (1cm * 3, -1cm * \d) {};
\node[world2] (s4\d) at (1cm * 4, -1cm * \d) {};
\node[world] (s5\d) at (1cm * 5, -1cm * \d) {};
}

\node[world2] (s35) at (1cm * 3, -1cm * 5) {};
\node[world2] (s45) at (1cm * 4, -1cm * 5) {};
\node[world] (s55) at (1cm * 5, -1cm * 5) {};

\node[dummy] (l51) at (1cm * 5 - 0.1cm, -1cm * 1 + 0.6cm) {$w^{\,n}_{d_1}$};
\node[dummy] (l55) at (1cm * 5 - 0.1cm, -1cm * 5 - 0.55cm) {$w^{\,n}_{d_n}$};
\foreach \i in {1,...,2}{
\node[dummy] (l\i1) at (1.3cm * \i, -1cm + 0.6cm) {$w^{\,\i}_{d_1}$};
\node[dummy] (l\i5) at (1.3cm * \i + 0.1cm, -1cm * 5 - 0.55cm) {$w^{\,\i}_{d_n}$};
}

\foreach \d in {1,...,2}{
\foreach \e in {1,...,2}{
\path[->] (s1\d) edge node{} (s2\e);
}}

\path[->,dotted]
(s0) edge node{} (s14)
(s0) edge node{} (s13);
\path[->] (s0) edge node{} (s15);

\foreach \i in {1,...,1}{
\foreach \d in {1,...,3}{
\pgfmathparse{int(\i + 1)}
\path[->,dotted] (s\i\d) edge node{} (s\pgfmathresult3);
}}

\foreach \d in {1,...,3}{
\foreach \e in {1,...,3}{
\path[->,dotted] (s2\d) edge node{} (s3\e);
\path[->,dotted] (s4\d) edge node{} (s5\e);
}}

\foreach \d in {4,...,5}{
\foreach \e in {4,...,5}{
\path[->,dotted] (s4\d) edge node{} (s5\e);
}}

\foreach \i in {1,...,2}{
\foreach \d in {4,...,5}{
\foreach \e in {4,...,5}{
\pgfmathparse{int(\i + 1)}
\path[->,dotted] (s\i\d) edge node{} (s\pgfmathresult\e);
}}}

\foreach \d in {1,...,2}{
\path[->] (s0) edge node{} (s1\d);
\path[->] (s5\d) edge node{} (sEnd);
}

\path[->]
(s15) edge node{} (s25)
(s55) edge node{} (sEnd)
(sEnd) edge node{} (sEndEnd)
(sEnd) edge[bend left=90,distance=4cm] node{} (s0)
(sEndEnd) edge [loop above] node{} (sEndEnd);

\path[->,dotted]
(s53) edge node{} (sEnd)
(s54) edge node{} (sEnd);

\path[->,dotted] (s13) edge node{} (s21)
(s13) edge node{} (s22)
(s13) edge node{} (s23);
\end{tikzpicture}

\vspace{-30pt}
\caption{Structure that models rectangle tilings as runs.}\label{fig:ltlmc-parapsp-tiling}
\end{figure}

The structure $\calS$ is shown in \Cref{fig:ltlmc-parapsp-tiling} and models tilings as follows: A run $\pi$ starts in $w_\text{left}$ and visits a row of $n$ worlds $w^i_d$. These worlds are the first row of the tiling. In every of the $n$ steps, $\pi$ may decide for any of the $\size{D}$ possible successors (which correspond to tiles). The back edge from $w_\text{right}$ to $w_\text{left}$ may be used then an arbitrary number of times, constructing a tiling consisting of many rows. The path may then enter the state $w_\text{end}$ and stay there forever.

We use the following formulas to check if the tiling is valid. First ensure that the complete first row has \emph{up} color $c_0$:
\[
\psi_\text{first} \dfn \bigwedge_{i = 1}^n \X^i (c_0)_u
\]
Check the neighbor to the right and below (if it is not the border):
\begin{align*}
\psi_{c,r} &\dfn \left[ c_r \rightarrow \X\left( q_\text{right} \lor c_l \right) \right]\\
\psi_{c,d} &\dfn \left[ c_d \rightarrow \X^{n+2}\left(q_\text{end} \lor c_u \right)  \right]
\end{align*}
The last row must exist and have \emph{down} color $c_1$:
\[
\psi_\text{last} \dfn \F \left[q_\text{left} \land \left(\X^{n+2} q_\text{end}\right)
\land \bigwedge_{i = 1}^{n} \X^i (c_1)_d  \right]
\]
The whole tiling is expressed by $\psi$:
\[
\psi \dfn \psi_\text{first} \land \psi_\text{last} \land
\neg \F \neg \bigwedge_{c \in C}\left(\psi_{c,r} \land \psi_{c,d} \right)
\]

Similar to \Cref{thm:ltl-mc-is-w1-hard} is the case that there is a valid $n\times m$-tiling if and only if a path starts in $w_\text{left}$ and satisfies $\psi$.

\begin{claim}
The formula $\psi$ has constant structural pathwidth.
\end{claim}
\begin{proof2}
We construct a path-decomposition $\calP$ of $\psi$ as follows.
Ignore the variables $(c_0)_u, (c_1)_d, q_\text{left}, q_\text{right}$ and $q_\text{end}$ as they can be added to every bag at the end, increasing every bag size only by five.

Now first process the formula $\psi_\text{first}$. It contains $n$ ``chains'' of $\X$-operators. For each such chain create a row of bags. For $j = 1,\ldots,n-1$ then add the $j$-th and the $(j+1)$-th $\X$-operator node to the $j$-th bag, so the edge in the syntactical structure between them is covered. Then the big conjunction is handled as in the proof of \Cref{thm:ltl-mc-is-w1-hard}, connecting the $n$ rows of bags to a single path, increasing the width by at most two.

The formulas $\psi_{c,r}$ have each constant length, so for every $c \in C$ put all of the nodes of subformulas of $\psi_{c,r}$ into a single bag and append it to $\calP$. This does not violate the path-decomposition rules as every variable $c_r, c_l$ appears only once in the whole formula $\psi$.

The length of $\psi_{c,d}$ depends on $n$, but the chain of $n+2$ $\X$-operators can be decomposed like in $\psi_\text{first}$, leading to $n+2$ new bags with each constant size.

The length of $\psi_\text{last}$ is again not constant; it contains a big conjunction of chains of $\X$-operators as well as another single chain with $q_\text{end}$ inside. Decompose the chains and the big conjunction as in $\psi_\text{first}$. The edges connecting them to the remaining number of constantly many $\land$-nodes and the $\F$ node can then be covered by adding the nodes to every bag, increasing the size only by a constant.

Finally for covering the whole formula $\psi$ in $\calP$, we need to insert the remaining small $\land$-operators, negations, and the $\F$ into every bag; and to decompose the big conjunction for every color $c$ append the bags of the $\psi_{c,d}$ formulas in the right order, again adding the small conjunction parts of the big conjunction. \qedBlack
\end{proof2}
%
From the above claims and \Cref{thm:ltlmc-in-pspace} the para-$\PSPACE$-completeness follows.
\end{proof}


\begin{theorem}\label{thm:ltl-mc-u-is-parapsp-hard}
\emph{$(\LTLMC(\U), \td + \pw_\varphi)$} and \emph{$(\LTLMC(\U), \td + \tw_\varphi)$} are \emph{para-$\PSPACE$}-complete.
\end{theorem}
\begin{proof}
For the Until operator the reduction from the proof of \Cref{thm:ltl-xf-parapspace-h} is possible in constant temporal depth. Similar to the proof of \Cref{thm:ltl-mc-u-w1-hard} adapt the structure $\calS$ and supplement the labeled color variables $c_u, c_d, c_l, c_r$ by their depth-aware versions, \ie, $c^i_u, c^i_d, c^i_l$ and $c^i_r$ for $1 \leq i \leq n$.

Modify the formulas as follows:
\begin{align*}
\psi_\text{first} &\dfn \left[q_\text{left} \lor (c_0)_u\right]\U q_\text{right}\\
\psi_\text{last} &\dfn \top \U \left[q_\text{left} \land [\left(q_\text{left} \lor (c_1)_d \right)\U q_\text{right}] \right]\\
\psi^i_{c,r} &\dfn c^i_r \rightarrow \left[c^i_r \U c^{i+1}_l\right]\\
\psi^i_{c,d} &\dfn c^i_d \rightarrow \left[c^i_d \U \left( \neg c^i_d \U (q_\text{end}\lor c^i_u) \right) \right]\\
\psi &\dfn \psi_\text{first} \land \psi_\text{last} \land \neg \bigg[\top \U \neg
\bigwedge_{i = 1}^n \bigwedge_{c \in C} \left( \psi^i_{c,r} \land \psi^i_{c,d} \right)\bigg]
\end{align*}

The variables $(c_0)_u, (c_1)_d, q_\text{left}$, and $q_\text{right}$ can again be added to every bag of a path-decomposition $\calP$ of $\psi$. The only part of $\psi$ that is not constant is the conjunction over the $n \cdot \size{C}$ subformulas $\psi^i_{c,r}$ and $\psi^i_{c,d}$. But each such subformula $\psi^i_{c,r}$, resp., $\psi^i_{c,d}$ can be covered by a single isolated bag: It has only a constant number of nodes and every occurring variable is either subformula-local or is already added to every bag.

Then it remains to decompose the big conjunctions which can be done in two steps. First connect the isolated bags for the inner conjunction and add the small conjunction nodes as needed. Then connect the resulting chains of length $\size{C}$ to finalize the path-decomposition $\calP$ that has a constant width.
\end{proof}

\newcommand{\paranph}{\text{para-}\NP\text{-h.}}
\newcommand{\paranpc}{\text{para-}\NP\text{-c.}}
\newcommand{\paraconph}{\text{para-}\coNP\text{-h.}}
\newcommand{\paraconpc}{\text{para-}\coNP\text{-c.}}

\newcommand{\parapspaceh}{\text{para-}\PSPACE\text{-h.}}
\newcommand{\parapspacec}{\text{para-}\PSPACE\text{-c.}}

\begin{figure}[!ht] \centering\small
\begin{tabular}{lll}\toprule
Problem $Q$ & Parameter $\kappa$ & \\ \midrule
$\LTLSAT(\cdot)$ & $\td$ & $(\td\; +\; )\tw_\varphi$/$\,\pw_\varphi$\\
\midrule
$\emptyset$ or $\X$
& $\paranpc$ 
& $\FPT$ \\
$\F$
& $\paranpc$ 
& $\W{1}$-h.\\
other 
& $\parapspacec$ 
& $\W{1}$-h. \\
\midrule
$\LTLMC(\cdot)$ & $\td$       & $\td + \Delta (+ \tw_\calA$/$\,\pw_\calA)$  \\
\midrule
$\X$   
& $\coW{\P}$, $\coW{1}$-h.
& $\FPT$        \\
$\F$   
& $\paraconpc$ 
& $\paraconpc$ \\
other 
& $\parapspacec$ 
& $\parapspacec$   \\
\midrule
& $\tw_\varphi$ / $\nVar$ & $\td + \tw_\varphi$/$\,\pw_\varphi$ \\
\midrule
$\X$   & $\paraconpc$
&$\coW{\P}$, $\coW{1}$-h.\\
$\F$   & $\paraconpc$
& $\coW{1}$-h. \\
$\F, \X$ & $\parapspacec$
& $\coW{1}$-h.\\
other   & $\parapspacec$
& $\parapspacec$\\
\midrule
&$\size{\varphi}$ & \\
\midrule
all  & $\FPT$
& \\
 \bottomrule
\end{tabular}
\caption{Overview over LTL parameterisations with their complexity.}
\end{figure}

\section{Conclusion}\label{sec:conclusion}

We showed that the intractable satisfiability and model checking problems of LTL cannot be tamed by applying the toolbox of parameterised complexity theory, at least not for the chosen well-known parameters. The model checking hardness is solely caused by the path semantics of LTL in branching Kripke structures. This conclusion holds both for the full set of LTL operators and for every operator fragment; the only exception is the case with only the $\X$ operator and a special parameterisation that allows a depth-bounded search tree algorithm.
As LTL is a special case of \CTLstarText, the hardness results immediately hold for \CTLstarText as well.

A future research possibility is to continue the search for tractable parameters of LTL. 
The ultimate goal is to find a non-trivial parameter, \ie, lower than the one of \citeauthor{lichPnue85}, that allows fixed-parameter tractability for all LTL operators. A first step could be LTL instances of simultaneously bounded formula treewidth and input structure treewidth.

\section*{Acknowledgment}
The authors acknowledge the support by DFG grant ME 4279/1-1.
Also the authors thank the anonymous referees for their helpful comments and hints regarding existing work, as well as colleagues for useful suggestions.

\bibliographystyle{plainnat}


\end{document}